\documentclass[12pt,reqno]{amsart}
\usepackage[samelinetheorem,notitle]{maherart}

\usepackage{xurl}
\usepackage{booktabs}

\usepackage{amsfonts,enumerate,bm,bbm,paralist,nicefrac,comment,float,thmtools,thm-restate,comment,multirow}
\usepackage[export]{adjustbox}
\makeatletter
\newcommand\fs@boxedtop
  {\fs@boxed
   \def\@fs@mid{\vspace\abovecaptionskip\relax}%
   \let\@fs@iftopcapt\iftrue
  }
\makeatother
\floatstyle{boxedtop}
\floatname{framedbox}{Procedure}
\newfloat{framedbox}{tbp}{lob}

\hypersetup{
    pdftitle =
        {Structural Advantages for Integrated Builders in MEV-Boost},
    pdfauthor =
        {Mallesh M Pai and Max Resnick},
    pdfsubject=
        {Bundle Sharing}
}

\usepackage{xparse}

\DeclareDocumentCommand\Pr{ m g }{%
    \ensuremath{   \IfNoValueTF {#2}
      {\mathbb{P}\left[{#1}\right]}
      {\mathbb{P}\left[{#1}\middle\vert{#2}\right]}%
    }
}
\DeclareDocumentCommand\E{ m g }{%
    \ensuremath{   \IfNoValueTF {#2}
      {\mathbb{E}\left[{#1}\right]}
      {\mathbb{E}\left[{#1}\middle\vert{#2}\right]}%
    }
}

\newif\ifdraft
\drafttrue

\begin{document}
\raggedbottom
\title{Structural Advantages for Integrated Builders in MEV-Boost}
\thanks{We thank Davide Crapis, Julian Ma, Barnab\`{e} Monnot, Alex Nezlobin and Thomas Thiery for helpful comments and suggestions. }
\author{Mallesh M. Pai$^\text{A}$} 
\address{$^\text{A}$\href{economics.rice.edu}{Rice University} and \href{www.mechanism.org}{Special Mechanisms Group}}
\author{Max Resnick$^\text{B}$}
\address{$^\text{B}$\href{www.mechanism.org}{Special Mechanisms Group}}

\email{mallesh.pai@mechanism.org, max.resnick@mechanism.org}

\begin{abstract}
Currently, over 90\% of Ethereum blocks are built using MEV-Boost, an auction that allows validators to sell their block-building power to builders who compete in an open English auction in each slot. Shortly after the merge, when MEV-Boost was in its infancy, most block builders were neutral, meaning they did not trade themselves but rather aggregated transactions from other traders. Over time, integrated builders, operated by trading firms, began to overtake many of the neutral builders. Outside of the integrated builder teams, little is known about which advantages integration confers beyond latency and how latency advantages distort on-chain trading. 

This paper explores these poorly understood advantages. We make two contributions. First, we point out that integrated builders are able to bid truthfully in their own bundle merge and then decide how much profit to take later in the final stages of the PBS auction when more information is available, making the auction for them look closer to a second-price auction while independent searchers are stuck in a first-price auction. Second, we find that latency disadvantages convey a winner's curse on slow bidders when underlying values depend on a stochastic price process that change as bids are submitted. 

\end{abstract}
\maketitle

\newpage

\section{Introduction}\label{sec:intro}
Nominally, new Ethereum blocks are produced by ordinary validators who, having been temporarily anointed as the \textit{proposer} by the protocol, gather transactions from the mempool (the set of publicly available transactions) and pack them together into the most valuable block they can produce.

Choosing a random validator to be the proposer is similar to choosing a random New Yorker to order all trades on the New York Stock Exchange for the next 12 seconds.\footnote{Ethereum builds a new block every 12 seconds so any trades that come in in the 12 second window between blocks are ordered subject to the discretion of the proposer.}  The typical New Yorker, lacking the institutional knowledge and infrastructure required, would not be able to extract much from his brief chronological hegemony over the exchange. On the other hand a high-frequency trading firm does have the expertise and therefore could extract much more from 12 seconds of chronological hegemony. This suggests a strategy for the chosen subject: instead of trying to extract the value himself, he should auction off the ordering rights to the highest bidder. Presumably there are multiple competing trading firms who would be able to extract significant value from this opportunity and so auction revenues will be high, and most of the value from the opportunity will be captured in the form of auction revenue and passed on to the selected pedestrian.%
\footnote{Indeed, one of the original motivations for PBS, which we introduce in the subsequent paragraph, is that a randomly chosen validator may have heterogeneous ability to extract value from their proposal rights, leading to eventual centralization in favor of more sophisticated validators. The existence of the auction ``levels the playing field'' among validators, which was viewed as desirable. Additionally, given the scarcity of blockspace, this was also seen as desirable from an efficiency perspective. See \cite{buterin} for more details.}

This strategy has been widely adopted by the Ethereum validator set. Today, over 90 percent of proposers elect to outsource the building of the block by auctioning it off to the highest bidder through MEV-Boost \citep{wahrstatter2023time}. The bidders in this auction are called \textit{builders}. As the name suggests, builders specialize in building valuable blocks, by any means necessary. Builders regularly expand the scope of their block-building beyond ordering transactions in the public mempool: they secure private order-flow, accept private bundles from MEV searchers, and even integrate with their own trading shops. 

The original conception of the PBS-auction did not involve these activities by searchers as far as we can tell. The idealistic version of the PBS auction was that each builder would be neutral (as the first few builders were and many still are). But 
\cite{gupta2023centralizing} revealed that several builders integrate with their own trading shops in order to more effectively capture CEX/DEX arbitrage. These `HFT' builders or integrated builders win far more blocks when the CEX price was volatile in the time since the preceding slot (when there was more CEX/DEX arbitrage available) relative to when the price was stable, confirming that these builders have a significant advantage in extracting CEX/DEX arbitrage. 

But why does integration confer such a sizable advantage? One obvious answer is that integration cuts down the communication time between the searcher and the builder, allowing for lower-latency strategies. Certainly, this is part of the reason that integration is important; however, it doesn't tell the whole story.  In this paper we propose two simple models to show how the latency advantage plays out in auction settings.

In the first, we consider a private values environment, i.e., one where the bidders' value for the good is not influenced by others' signals. This is the case, for example, if bidders' values come purely from the associated tips of the transactions in their block: the fact that other builders have a different value just reflects the fact that they constructed a different block, and does not change their value for winning the PBS auction with their own block. 
The PBS auction proceeds in two phases: in the first phase, \textit{block construction}, searchers submit bundles to builders, and builders construct the blocks that maximize their profits. In the second stage builders compete with each other to see which of their blocks wins. Independent searchers must take profits in the block construction phase, which means they must shade their bids (relative to the value of the bundle/ MEV opportunity) in order to retain profit for themselves. Since they do not know the value of the opportunities others have found, they do not know precisely how much to shade to get included in the next block. They instead choose a shading amount that maximizes their expected profit, i.e. their profit upon inclusion times the probability of inclusion. The integrated-searchers on the other hand know how much they expect to make if their bundle lands in the first stage. They can then decide how much profit to take during the second stage of the auction, which is an English auction. Since the first stage occurs earlier, and is typically sealed bid, independent searchers are acting in a pay-your-bid regime, whereas integrated searchers are bidding in something akin to a second-price regime.

The second model considers a common-value setting. This for example is the case for a top-of-block CEX/DEX arb (i.e., where the competition is for the right to trading with an AMM pool on-chain, using information from a continuous time CEX).%
\footnote{Such arbs are a major source of AMM liquidity provider (LP) losses, see e.g. \cite{milionis2022automated,milionis2023automated} for theoretical foundations and \cite{milionis2022quantifying} for some estimates.} In this case latency advantages manifest as an information advantage, i.e. the ability to bid later in the auction. 

The basic upshot of our results is that in either case, integrated/ low-latency bidders are advantaged. Our results therefore provide additional reasons for the rise of integrated builders (originally documented in \cite{gupta2023centralizing}).

\section{Private Values}

There are $n = n_A + n_B$ builders who are builders in the PBS auction. Of these $n_A$ are integrated builders and $n_B$ are independent/ non-integrated builders. In this section we assume independent private values: every integrated builder has a value that is drawn from a distribution with CDF $F_A$ and density $f_A$, and every non-integrated builder has a value drawn from a distribution with CDF $F_B$ and density $f_B$. All draws are independent.

We consider the following hybrid auction:
\begin{enumerate}
\item There is a single item for sale.
\item All bidders simultaneously submit bids.
\item The highest bidder wins the item. Payments are defined as follows:
\begin{enumerate}
    \item If the winner is an integrated builder, then they pay the next-highest bid.
    \item If the winner is instead an independent builder, they pay their own bid. 
\end{enumerate}
\end{enumerate}

In short, the integrated builders compete in, what is in their view, a second-price auction while the non-integrated builders compete in a first-price auction. 

Before we proceed let us make explicit how the model above corresponds to the setting of interest. Suppose there are $n$ searchers. Each of these searchers finds single profitable opportunity (for example, this could be a sandwich on an existing transaction in the mempool). We assume that each searcher's opportunity is unique.\footnote{In practice multiple searchers may find the same or related opportunities, this could be modeled as correlated instead of independent values---we believe this would complicate the model without adding any extra insight.} Each searcher then needs to submit their opportunity to a builder to be included. Integrated searcher-builders simply submit their opportunity to their own integrated builder. Independent searchers submit their opportunity to neutral builders. We assume that each independent searcher sends to a different netural builder, in line with the simplifying assumption made in \cite{feldman2010auctions} in the context of Ad Exchanges.\footnote{Relaxing this assumption would lead to correlated values among the neutral builders, which in our opinion would complicate the model without adding any extra insight.} At the stage of submitting to these neutral builders, the independent searchers must choose how much they should pass along as a fee to the builder, for instance in the form of a bundle tip. A higher tip increases the chances of inclusion on the chain, but reduces the net profit  of the searcher.  Subsequently, the builders bid in the PBS auction. We assume that there are no other transactions so the block that the winning bidder would propose is solely the opportunity from their corresponding searcher. The value of a builder for winning a block is therefore the value of the underlying opportunity for an integrated searcher-builder, or the associated tip for a neutral builder. Since the PBS auction is an English auction, the outcome is the same as a standard sealed-bid second-price auction. From the perspective of an independent searcher, however, they have to submit their opportunity before knowing the bidding, and therefore their builder tip reflects a bid into a pay-your-bid auction. 

So now to solve this auction. 
\begin{observation}\label{obs:integrated}
    It is a weakly dominant strategy for integrated builders to bid their true values in this hybrid auction, so we restrict attention to these strategies for the integrated builder.
\end{observation}

\begin{lemma}
    From the point of view of the non-integrated builders, the auction is equivalent to a first-price auction  of $n_B$ agents where there is a random, secret reserve price distributed according to the CDF $F_A^{n_A}$.
\end{lemma}
\begin{proof}
    By Observation \ref{obs:integrated}, integrated bidders bid their values in this hybrid auction. Therefore the highest of these is distributed according to the CDF $F_A^{n_A}$. Since the auction is pay-as-bid for the non-integrated bidders, from their perspective this is a first-price auction with $n_B$ total bidders. However to win the object, the winning bidder also needs to be higher than any of the integrated builders. The latter therefore act as a secret reserve price distributed according to $F_A^{n_A}$.
\end{proof}

We can now characterize the symmetric Bayes-Nash equilibrium bidding strategy of non-integrated buyers. Let $\sigma(\cdot)$ denote bidding strategy of the non-integrated buyers, i.e., a bidder of value $v$ bids $\sigma(v)$. It is immediate that $\sigma(\cdot)$ is a strictly increasing function. Then we must have:
\begin{lemma}\label{lem:ni-eqbm}
The symmetric Bayes-Nash equilibrium among non-integrated bidders $\sigma(\cdot)$ solves:
\begin{align} \label{eqn:ni-eqbm}
    \sigma(v) = v - \frac{\int_{0}^v  F_B^{n_B-1}(t) F_A^{n_A}(\sigma(t)) dt}{F_B^{n_B-1}(v) F_A^{n_A}(\sigma(v))}.
\end{align}
\end{lemma}

\begin{figure}
    \centering
    \caption{Strategies for Beta distributed private values with $N_A = N_B = 3$}
    \includegraphics[width = \textwidth]{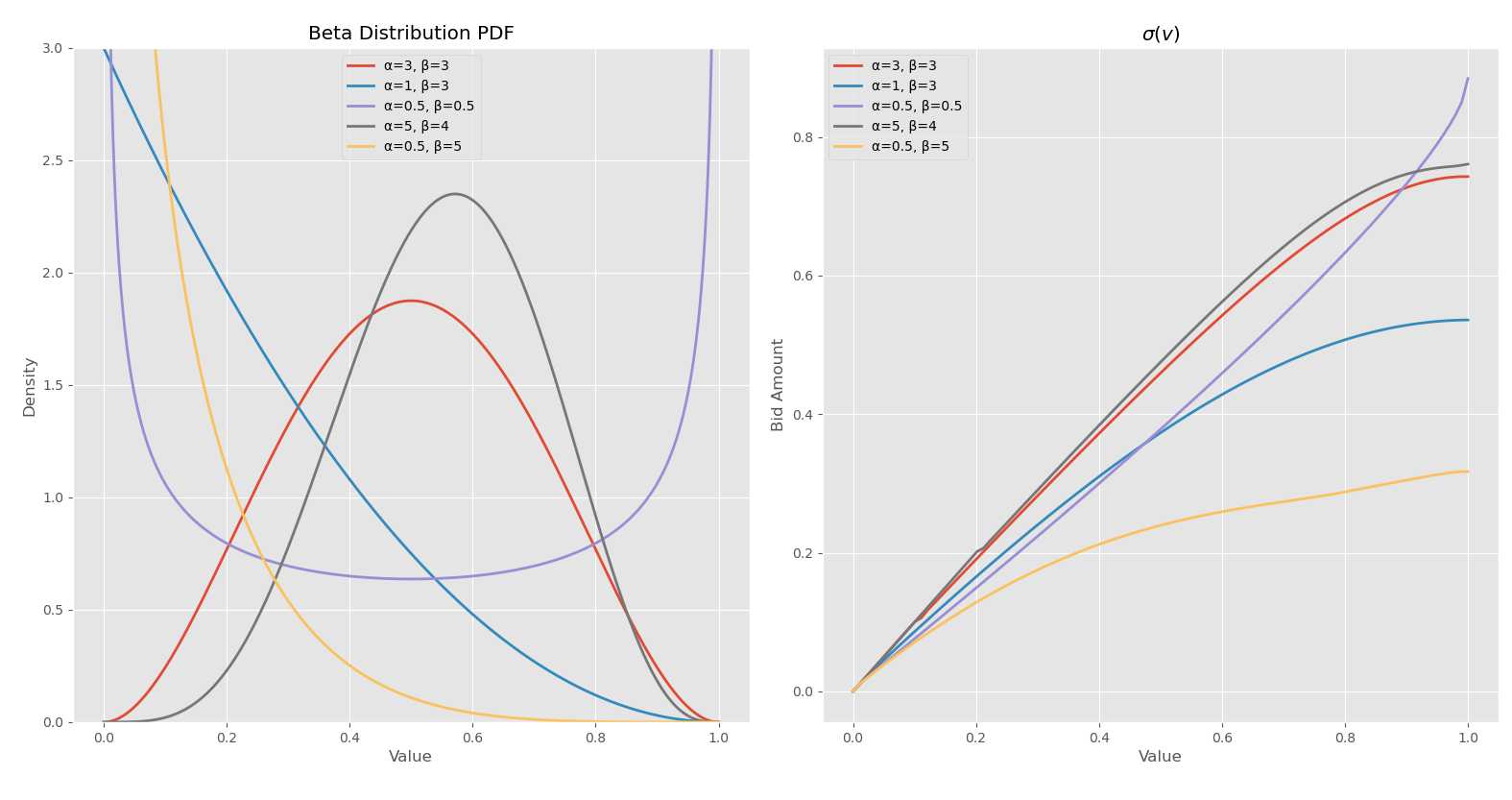}
    \label{fig:beta_dist}
\end{figure}
\begin{proof}
To see this, observe that if all non-integrated bidders bid according to $\sigma(\cdot)$, then a bidder of value $v$ wins with probability:
\begin{align*}
    x(v) = F_B^{n_B-1}(v) F_A^{n_A} (\sigma(v)). 
\end{align*}
Here the first term is the probability of beating all other non-integrated bids (since $\sigma$ is strictly increasing) and the latter is the probability of the bid beating the secret-reserve price/ all $n_A$ integrated bidders' bids. 

Then by revenue equivalence (see e.g. \cite{krishna2009auction}), this bidder's interim expected surplus in this auction must equal 
\begin{align*}
    S(v) = \int_0^v x(t) dt. 
\end{align*}
However, by bidding $\sigma(v)$ in a first price auction where you win with probability $x(v)$, the interim expected surplus must equal $(v - \sigma(v)) x(v)$. 

Equating the two we have \eqref{eqn:ni-eqbm} as desired. \end{proof}

\subsection{Analytical solutions}
To see analytical solutions for this let's firstly assume that $n_B=1$ and that $F_A$ is the uniform distribution on [0,1], i.e., $F_A(t) = t \iff t \in [0,1]$. Figure \ref{fig:beta_dist} displays numerical results for a range of Beta-distributions. 

In this case \eqref{eqn:ni-eqbm} reduces to:
\begin{align*}
    &\sigma(v) = v - \frac{\int_0^v \sigma^{n_A}(t) dt}{ \sigma^{n_A}(v)}\\
\implies & \sigma^{n_A+1}(v) - v \sigma^{n_A}(v) + \int_0^v\sigma^{n_A}(t) dt =0.\\
\intertext{Differentiating, we have:}
& (n_A+1) \sigma^{n_A}(v) \sigma'(v) - n_A v \sigma^{n_A-1}(v) \sigma'(v)=0,\\
\implies& \sigma'(v) \sigma^{n_A-1}(v) \left( (n_A+1) \sigma(v) - n_A v \right) =0\\
\implies& \sigma(v) = \frac{n_A}{n_A+1} v
\end{align*}

So in this case, we have that:
\begin{proposition}
If there is a single non-integrated builder, it bids $\frac{n_A}{n_A+1}$ of its value in the auction. As a result its equilibrium surplus in the auction is:
\begin{align*}
    &S(v) = \left(\frac{n_A}{n_A+1}\right)^{n_A} \frac{v^{n_A+1}}{n_A+1},
\intertext{whereas if this bidder was also integrated its surplus in the auction would be}
& S(v) = \frac{v^{n_A+1}}{n_A+1}.
\end{align*}
In other words being the sole non-integrated builder costs it a fraction  $\left(\frac{n_A}{n_A+1}\right)^{n_A}$ of its surplus relative to if it had been an integrated builder. 
\end{proposition}

Now let's suppose that $n_B= n$ for $n>1$. We'll denote the equilibrium in this case by $\sigma_n$. Assume also that $F_B$ is also the uniform distribution on $[0,1]$

In this case, substituting into \eqref{eqn:ni-eqbm}, we have that $\sigma_n$ is the solution to
\begin{align*}
    &\sigma_n(v) = v - \frac{\int_0^v t^{n-1} \sigma_n^{n_A}(t) dt}{v^{n-1} \sigma_n^{n_A}(v)},\\
   \implies &  v^{n-1} \sigma_n^{n_A+1}(v) - v^n \sigma_n^{n_A} (v) + \int_0^v t^{n-1} \sigma_n^{n_A}(t) dt =0,\\
\intertext{Differentiating wrt v, we have that:}
&(n-1) v^{n-2} \sigma_n^{n_A+1}(v) + (n_A+1)  v^{n-1} \sigma_n^{n_A}(v) \sigma_n'(v) -  (n-1) v^{n-1} \sigma_n^{n_A} (v) - n_A v^n \sigma_n^{n_A-1}(v) \sigma'_n(v) =0.
\intertext{Dividing throughout by $v^{n-2}$ gives us}
\implies& (n-1) \sigma_n^{n_A+1}(v) + (n_A+1)  v \sigma_n^{n_A}(v) \sigma_n'(v) -  (n-1) v \sigma_n^{n_A} (v) - n_A v^2 \sigma_n^{n_A-1}(v) \sigma'_n(v) =0.
\intertext{Dividing throughout by $\sigma_n^{n_A-1}$ gives us}
\implies &  (n-1) \sigma_n^{2}(v) + (n_A+1)  v \sigma_n(v) \sigma_n'(v) -  (n-1) v \sigma_n (v) - n_A v^2  \sigma'_n(v) =0.
\intertext{Collecting terms we have that}
& ((n_A+1)  v \sigma_n(v) - n_A v^2) \sigma_n'(v) = (n-1)  \sigma_n(v) (v  - \sigma_n(v))\\
 \implies& \sigma_n'(v) = \frac{(n-1) \sigma_n (v) (v - \sigma_n(v) )}{v ((n_A+1)   \sigma_n(v) - n_A v)} 
\end{align*}
Since $\sigma_n' >0$, so that the numerator and denominator are both positive.  We can conclude that:
\begin{proposition}
The equilibrium bids of the $n$ non-integrated bidders as a function of their value, $\sigma_n(v)$, satisfies:
    \begin{align*}
    \frac{n_A}{n_A+1} v  \leq \sigma_n (v) \leq v.
\end{align*}
\end{proposition}

It follows straightforwardly, by an analogous appeal to revenue equivalence, that each of these non-integrated bidders still has a lower expected surplus than if they were bidding truthfully.

\section{Common Values}

We now consider a setting where the object for sale has a common value, i.e. the value of the object is the same regardless who eventually wins. This may be the case if the object for sale is the top slot of a block, designated as e.g. a CEX-DEX arb or a DeFi liquidation. The realized value would then be the same regardless of who wins it. Such a common value model was proposed in \cite{ma2023dynamic}.

We consider the case of two kinds of bidders as before, fast bidders and slow bidder. The slow bidders bid at time 0, at which time the present value of the object is $v_0$. The process follows a martingale, i.e., we assume that $t$ seconds after this the process is worth $v_t = (\exp m_t) v_0$ where $m_t \sim N(\mu t, \frac{\sigma^2}{2} t) $ and $\mu = -\frac{\sigma^2}{2}.$ In particular the latter implies that $\mathbb{E}[v_{t+\delta}|v_t]=v_{t}$ for any $t, \delta \geq 0$. 

\subsection{An Unraveling Result}
We first study the case where the fast bidder bids exactly $\Delta$ seconds after the slow bidders, i.e. they have superior information. The proposition below shows the complete unraveling of the auction. 

\begin{proposition} \label{prop:imp}
    In the auction described above, the slow bidders always bid 0 and the fast bidder always wins. 
\end{proposition}
\begin{proof}
    To see this suppose a slow bidder bids any $b >0$. 
    
    The fast bidder continues to see the evolution of the process and bids at time $\Delta$ after the slow bidder, at which point the object is worth $v_{\Delta}$. The fast bidder has a simple optimal strategy--- if $v_{\Delta}>b $ then outbidding the slow bidder has positive expected value. If $v_{\Delta}\leq b$ then the outbidding the slow bidder has negative expected value. 

    Therefore the fast bidder should outbid the slow bidder if and only if $v_{\Delta} >b$. 

    Therefore the slow bidder wins the object if and only if $v_{\Delta} \leq b$. But in this case, the expected value of the object conditional on winning is $\mathbb{E}[v_{\Delta}| v_\Delta \leq b] <b$. As a result bidding any $b>0$ has strictly negative expected value for the slow bidder. Therefore all slow bidders should bid 0, and be outbid with probability equal to 1 by a fast bidder.
\end{proof}

\subsection{The PBS Candlestick}
Validators are not required to call the block exactly at the start of the slot \citep{schwarzschilling2023time}, network delays are never a certainty, and clock synchronization can vary between validators. What this means is that, in practice, a speed advantage may not manifest itself deterministically. A slow bidder may simply have fewer opportunities to update their bid in the PBS auction (as a result of their speed, or lack thereof). We model this as a candlestick auction.%
\footnote{A candlestick auction is one where bids are accepted until a (real-world) candle runs out, or some other stochastic process terminates. This was intended to ensure that no one could know exactly when the auction would end and make a last-second bid. See \url{https://en.wikipedia.org/wiki/Candle_auction} for further details. }

To capture the speed advantage of the fast bidder, we consider the following simplified model:
\begin{enumerate}
    \item The slow bidders submit sealed bids at time $0$.  
    \item The fast bidder has an option to bid at time $\Delta$ with probability $p \in [0,1]$. 
\end{enumerate}
We assume there are at least two slow bidders.%
\footnote{If there is a single slow bidder, then the optimal strategy for them is trivially to bid $0$ and win the auction only when the fast bidder does not have a revision opportunity.}

This model is a simplified version of a ``revision game'' as originally proposed in \cite{kamada2020revision}. Slow bidders can submit a first bid, but then the fast bidder has a stochastic opportunity to revise their bid, while the slow bidder does not have such an opportunity. The parameter $p$ therefore captures the ``speed'' of the fast bidder, i.e., a higher $p$ corresponds to a bidder who can respond more often to the slow bidder's bid even if they get outbid initially. Intuitively $p=0$ corresponds to a standard sealed-bid common value auction (which in this case will conclude at price $v_0$), while $p=1$ corresponds to the unraveling case above. A more general model where slow bidders have stochastically fewer revision opportunities than fast bidders is left to future work. 

\begin{proposition}
    The winning slow bidder in the period 0 auction bids the largest $b_0^S \in [0,v_0]$ solving: 
    \begin{align}\label{eqn:b0s}
    (1-p) (v_0 - b_0^S) + p P(v_\Delta < b_0^S) (E(v_\Delta | v_\Delta < b_0^S) - b_0^S) =0.
\end{align}
They win the item with probability $p P(v_\Delta < b_0^S) + (1-p).$
\end{proposition}
\begin{proof}
    Let us denote the winning slow bidder's bid  bid $b_0^S$. 

    The fast-bidder's move in the subgame after the period $0$ auction \emph{if they get a revision opportunity} is straightforward: whenever the current value at time $\Delta$, $v_\Delta > b_0^S$ then the fast-bidder updates their bid to $b_0^S$. If instead $v_\Delta < b_0^S$, then they choose lose the auction (i.e. let the slow bidder win despite having a a revision opportunity).

    Therefore, at time $0$, the slow bidder understands that if they have the higher bid at time $0$, there are two possibilities:
    \begin{enumerate} 
    \item They lose money whenever the fast bidder gets an opportunity to update their bid. In particular, in expectation, their net expected profit in this case is:
    \begin{align*}
        P(v_\Delta < b_0^S) (E(v_\Delta | v_\Delta < b_0^S) - b_0^S) 
    \end{align*}
    \item If the fast bidder does not get an opportunity to update their bid, the slow bidder's expected profit in this case is:
    \begin{align*}
        (v_0 - b_0^S).
    \end{align*}
    \end{enumerate}
Note that by assumption the probability of the latter event is $1-p$ and the former is $p$. Consider the solution to:
\begin{align}
    (1-p) (v_0 - b_0^S) + p P(v_\Delta < b_0^S) (E(v_\Delta | v_\Delta < b_0^S) - b_0^S) =0.
\end{align}

Note that the left hand side is strictly positive at $b_0^S=0$ and strictly negative at $b_0^S = v_0$, which implies that there exists at least one solution by the intermediate value theorem. The largest solution to this is the largest amount a slow bidder can bid without losing money. 

We claim that any equilibrium must involve at least one slow bidder bidding the largest $b_0^S$ solving \eqref{eqn:b0s} in period $0$. To see why note that the since the slow bidders are symmetric they must have 0 profit in equilibrium. Now suppose the winning bid is something other than the largest root. Since the left hand side is positive at $b_0^S=0$ and negative at $b_0^S=v_0$, by continuity, if there are multiple roots and the winning slow bidder's bid is one of the smaller roots, then there exists a larger bid such that the left hand side of \eqref{eqn:b0s} evaluates to strictly positive, so this cannot be an equilibrium. 
\end{proof}

Therefore, stochastic opportunities for the fast bidder to outbid the slow bidder lessens the adverse selection faced by the latter, and therefore allow for positive bids and positive probability of winning for the fast bidder.

\section{Conclusion}
As has been documented in multiple venues, the builder market in Ethereum increasingly favors integrated builder-searchers. A priori, one might think that there is no reason for searchers to integrate with builders, and indeed that competitive searchers might be better served focusing on searching and simply letting builders compete for their orders. While it is understood that integration confers and advantage to integrated builders relative to non-integrated builders, little is understood about exactly how. This paper sheds a first light on how latency advantages play out in an auction setting, both in the case of private value and common values. We leave to future work quantifying the source of the advantage, i.e., additional time to construct blocks etc versus advantages in the auction.  

\bibliographystyle{econometrica}
\bibliography{onchain-auctions}

\end{document}